\documentclass[12pt,a4paper]{article}
\usepackage[top=2.1cm,bottom=2.1cm,left=2.1cm,right=2.1cm]{geometry}
\usepackage{latexsym,amsmath,amsfonts,amssymb,amsthm,amscd,mathrsfs}

\usepackage{accents}

 \usepackage[unicode,colorlinks,plainpages=false,hyperindex=true,bookmarksnumbered=true,bookmarksopen=false,pdfpagelabels]{hyperref}
 \hypersetup{urlcolor=cyan,linkcolor=blue,citecolor=red,colorlinks=true}

\newcommand{\be}{\begin{equation}}
\newcommand{\ee}{\end{equation}}
\newcommand{\bea}{\begin{eqnarray}}
\newcommand{\eea}{\end{eqnarray}}

\numberwithin{equation}{section}
\newcounter{thmcounter}
\numberwithin{thmcounter}{section}

\theoremstyle{definition}
\newtheorem*{acknowledgements}{Acknowledgements}
\newtheorem{definition}[thmcounter]{Definition}
\newtheorem{remark}[thmcounter]{Remark}
\theoremstyle{plain}

\newtheorem{lemma}[thmcounter]{Lemma}
\newtheorem{proposition}[thmcounter]{Proposition}
\newtheorem{theorem}[thmcounter]{Theorem}

\def\1{{\boldsymbol 1}}                     %
\def\cD{{\mathcal D}}                       %
\def\cH{{\mathcal H}}                       %
\def\tr{\mathrm{tr}}                        %
\def\diag{\mathrm{diag}}                    %
\def\ri{{\rm i}}                            %
\def\C{\mathbb{C}}                          %
\def\N{\mathbb{N}}                          %
\def\R{\mathbb{R}}                          %
\def\T{\mathbb{T}}                          %
\def\cF{{\mathcal F}}                       %
\def\reg{\mathrm{reg}}                      %
\def\red{\mathrm{red}}                      %
\def\Ad{{\mathrm{Ad}}}                      %
\def\id{{\mathrm{id}}}                      %
\def\dt {\left.\frac{d}{dt}\right|_{t=0}}   %
\def\fM{\mathfrak{M}}                       %
\def\cN{{\mathcal N}}                       %
\def\GL{{\rm GL}(n,\C)}                     %
\def\gl{{\rm gl}(n,\C)}                     %
\def\cG{{\mathcal G}}                       %
\def\dz {\left.\frac{d}{dz}\right|_{z=0}}   %
\def\Holm{\mathrm{Hol}(\fM)}                %
\def\Hol{\mathrm{Hol}}                      %
\def\cR{{\mathcal R}}                       %
\def\ad{\mathrm{ad}}                       %
\def\I{{\mathbb{I}}}                        %
\def\half{\frac{1}{2}}                      %
\def\sgn{\mathrm{sgn}}

\begin{document}

\begin{center}
{\Large\bf
Bi-Hamiltonian structure of spin Sutherland  models: the holomorphic case}
\end{center}

\medskip
\begin{center}
L.~Feh\'er${}^{a,b}$
\\

\bigskip
${}^a$Department of Theoretical Physics, University of Szeged\\
Tisza Lajos krt 84-86, H-6720 Szeged, Hungary\\
e-mail: lfeher@physx.u-szeged.hu

\medskip
${}^b$Department of Theoretical Physics, WIGNER RCP, RMKI\\
H-1525 Budapest, P.O.B.~49, Hungary\\
\end{center}

\medskip
\begin{abstract}
We construct a bi-Hamiltonian structure for the holomorphic spin Sutherland hierarchy
based on collective spin variables.  The construction relies
on Poisson reduction of a bi-Hamiltonian structure on the holomorphic cotangent bundle
of $\GL$, which itself arises from the canonical symplectic structure
and the Poisson structure of the Heisenberg double of the
standard  $\GL$ Poisson--Lie group.
The previously obtained bi-Hamiltonian structures of the hyperbolic and trigonometric real forms
are recovered on real slices of the holomorphic spin Sutherland model.
\end{abstract}

{\linespread{0.8}\tableofcontents}

\newpage
\section{Introduction}
\label{sec:I}

The theory of integrable systems is an interesting field of mathematics motivated by influential examples
of exactly solvable models of theoretical physics. For reviews,  see e.g. \cite{A,BBT,RBanff,SurB}.
There exist several approaches to integrability. One of the most popular ones in connection with
classical integrable systems is the bi-Hamiltonian method, which
originates from the work of Magri \cite{M} on the KdV equation, and plays an important role in
generalizations of this infinite dimensional bi-Hamiltonian system \cite{DeKV}.
As can be seen in the reviews,
among finite dimensional integrable systems the central position is occupied by Toda models
and the models that carry the names
of Calogero, Moser, Sutherland, Ruijsenaars and Schneider.
The Toda  models have a relatively well developed bi-Hamiltonian description \cite{SurB}.
The Calogero--Moser type models and their generalizations are much less explored from this point of view,
except for the rational Calogero--Moser model  \cite{AAJ,Bar,FaMe}.
In our recent work \cite{FeNon,FeLMP}  we made a step towards improving this situation by providing a bi-Hamiltonian
interpretation for a family of spin extended  hyperbolic and trigonometric Sutherland models.
In these references we investigated real-analytic Hamiltonian systems, and here wish to extend the pertinent results
to the corresponding complex holomorphic case.

Specifically, the aim of this paper is to  derive  a bi-Hamiltonian description
for the hierarchy of holomorphic evolution equations of the form
\be
\dot{Q}= (L^k)_0 Q,
\qquad
\dot{L} = [ \cR(Q)(  L^k), L], \qquad \forall k\in \N,
\label{I1}\ee
where $Q$ is an invertible complex diagonal matrix of size $n\times n$,
$L$ is an arbitrary $n\times n$ complex matrix,
and the subscript $0$ means diagonal part.
The eigenvalues $Q_j$ of $Q$ are required to be distinct, ensuring that the formula
\be
\cR(Q) := \frac{1}{2} (\Ad_Q + \id)(\Ad_Q - \id)^{-1}, \quad \hbox{with}\quad \Ad_Q(X):= Q X Q^{-1},
\label{I2}\ee
gives a well-defined linear operator on the off-diagonal subspace of $\gl$.
By definition, $\cR(Q)\in \mathrm{End}(\gl)$ vanishes on the diagonal matrices,
and one can recognize it as the basic dynamical $r$-matrix \cite{BDF,EV}.
Like in the real case \cite{FeNon}, it follows from the classical dynamical Yang--Baxter
equation satisfied by $\cR(Q)$  that
the evolutional derivations (\ref{I1}) pairwise commute \emph{if} they act on
such `observables' $f(Q,L)$  that are invariant with respect to conjugations of $L$
by invertible diagonal matrices.

The  system (\ref{I1}) has a well known interpretation as a holomorphic Hamiltonian system \cite{LX}.
This arises from  the parametrization
\be
L= p + (\cR(Q) + \frac{1}{2}\id)(\phi),
\label{I5}\ee
where $p$ is an arbitrary diagonal and $\phi$ is an arbitrary off-diagonal matrix.
The diagonal entries $p_j$ of $p$ and $q_j$ in $Q_j = e^{q_j}$ form canonically
conjugate pairs. The vanishing of the diagonal part of $\phi$ represents a constraint
on the linear Poisson space $\gl$,  and this is responsible for the gauge transformations acting
on $L$ as conjugations by diagonal matrices.
The $k=1$ member of the hierarchy (\ref{I1}) is governed by the
 standard spin Sutherland Hamiltonian
\be
H_{\mathrm{Suth}}(Q,p,\phi)= \frac{1}{2} \tr \left(L(Q,p,\phi)^2\right)   =
\frac{1}{2} \sum_{i=1}^n p_i^2 + \frac{1}{8} \sum_{k\neq l}
\frac{ \phi_{kl} \phi_{lk}}{\sinh^2\frac{q_k - q_l}{2}}.
\label{I7}\ee
For this reason, we may refer to  \eqref{I1} as the holomorphic spin Sutherland hierarchy.

It is also known (see e.g. \cite{Res1}) that the holomorphic spin Sutherland hierarchy is a reduction
of a natural integrable system on the
cotangent bundle $\fM:=T^* \GL$ equipped with its canonical symplectic form.
Before reduction, the elements of $\fM$ can be represented by pairs $(g,L)$, where $g$ belongs
to the configuration space and $(g,L) \mapsto L$  is the moment map for left-translations.
The Hamiltonians $\tr(L^k)$  generate an integrable system on $\fM$, which reduces to the
spin Sutherland system by keeping only the observables that are invariant under
simultaneous conjugations of $g$ and $L$ by arbitrary elements of $\GL$.
This procedure is called Poisson reduction.
We shall demonstrate that \emph{the unreduced integrable system on $\fM$
possesses a bi-Hamiltonian structure that descends to a bi-Hamiltonian structure of the spin
Sutherland hierarchy via the Poisson reduction}.

A holomorphic (or even a continuous) function on $\fM$ that is invariant
under the $\GL$ action \eqref{act1}
can be recovered
from its restriction to $\fM_0^\reg$, the subset of $\fM$ consisting of the pairs $(Q,L)$ with diagonal
and regular $Q\in \GL$.
Moreover, the restricted function
inherits invariance with respect to the normalizer of the diagonal subgroup $G_0< \GL$, which includes $G_0$.
This explains the gauge symmetry of the hierarchy \eqref{I1}, and lends justification
to the restriction on the eigenvalues of $Q$.

The bi-Hamiltonian structure on $\fM$ involves in addition to the canonical Poisson bracket
associated with the universal cotangent bundle symplectic form
another one that we construct from Semenov-Tian-Shansky's Poisson bracket of the Heisenberg
double of $\GL$ endowed with its standard Poisson--Lie group structure \cite{STS}.
Surprisingly, we could not find it in the literature that
the canonical symplectic structure of the cotangent bundle $\fM$ can be complemented to a
bi-Hamiltonian structure in this manner. So this appears to be a novel result, which is given
by Theorems \ref{Th:2.1}, \ref{Th:2.2+} and Proposition \ref{Prop:2.3} in Section \ref{sec:E}. The actual derivation
of the second Poisson bracket \eqref{E9} is relegated to an appendix.
The heart of the paper is Section \ref{sec:F}, where we derive the bi-Hamiltonian structure
of the system \eqref{I1} by Poisson reduction. The main results are encapsulated by
Theorem \ref{Th:3.4} and Proposition \ref{Prop:3.6}.
The first reduced Poisson bracket \eqref{red1} is associated with the spin Sutherland
interpretation by means of the parametrization \eqref{I5}.
The formula of  the second reduced Poisson bracket is given by equation \eqref{red2}.
After deriving the holomorphic bi-Hamiltonian structure in Section \ref{sec:F},
we shall explain in Section \ref{sec:R} that it  allows us to recover the bi-Hamiltonian structures of the
hyperbolic and trigonometric real forms derived earlier by different means \cite{FeNon,FeLMP}.
In the final section, we  summarize
the main results once more, and highlight a few open problems.

\section{Bi-Hamiltonian hierarchy on the cotangent bundle}
\label{sec:E}

Let us denote  $G:= \GL$ and equip its Lie algebra $\cG:=\gl$ with the trace form
\be
\langle X, Y \rangle := \tr(XY),
\quad \forall X, Y \in \cG.
\label{E1}\ee
This is a non-degenerate, symmetric bilinear form that enjoys the invariance property
\be
\langle X, Y \rangle = \langle \eta X \eta^{-1}, \eta Y \eta^{-1}\rangle, \qquad \forall \eta \in G,\, X,Y\in \cG.
\label{trinv}\ee
Any $X\in \cG$ admits the unique decomposition
\be
X = X_> + X_0 + X_<
\label{E2}\ee
into strictly upper triangular part $X_>$, diagonal part $X_0$, and strictly lower triangular part $X_<$.
Thus $\cG$ is the vector space direct sum of the corresponding subalgebras
\be
\cG = \cG_> + \cG_0 + \cG_<.
\label{RA}\ee
We shall use
the standard solution of the modified classical Yang--Baxter equation
on $\cG$,  $r\in \mathrm{End}(\cG)$ given by
\be
r(X):=   \frac{1}{2} (X_> - X_<),
\label{E3}\ee
and define also
\be
r_\pm := r \pm \frac{1}{2} \mathrm{id}.
\label{rpm}\ee

Our aim is to present two holomorphic Poisson structures on  the complex manifold
\be
\fM:= G \times \cG = \{ (g,L)\mid g\in G, \, L\in \cG\}.
\label{E4}\ee
Denote $\Holm$ the commutative algebra of holomorphic functions on $\fM$.
For any$F \in \Holm$, introduce the $\cG$-valued derivatives $\nabla_1 F$, $\nabla_1'F$  and $d_2 F$ by
the defining relations
\be
\langle \nabla_1 F(g,L), X\rangle = \dz F(e^{zX} g, L),
\quad
\langle \nabla_1' F(g,L), X\rangle = \dz F(ge^{zX}, L)
\label{E5}\ee
and
\be
\langle d_2 F(g,L), X \rangle = \dz F(g, L + z X),
\label{RB}\ee
where $z$ is a complex variable and $X\in \cG$ is arbitrary.
In addition, it will be convenient to define the $\cG$-valued
functions $\nabla_2 F$ and $\nabla_2' F$ by
\be
\nabla_2 F(g,L) := L d_2 F(g,L),
\qquad
\nabla_2' F(g,L) := (d_2F(g,L)) L.
\label{E7}\ee
Note that
\be
\nabla_1' F(g,L) =  g^{-1} (\nabla_1 F(g,L)) g,
\label{E6}\ee
and a similar relations holds between $\nabla_2 F$ and $\nabla_2' F$ whenever $L$ is invertible.

\begin{theorem} \label{Th:2.1}
For holomorphic functions $F, H\in\Holm$, the following formulae define two
Poisson brackets:
\be
\{ F,H\}_1(g,L) =   \langle \nabla_1 F, d_2 H\rangle - \langle \nabla_1 H, d_2 F \rangle + \langle L, [d_2 F, d_2 H]\rangle,
\label{E8}\ee
and
\bea
\{ F, H\}_2(g,L)  &=&
 \langle r \nabla_1 F, \nabla_1 H \rangle - \langle r \nabla'_1 F, \nabla'_1 H \rangle
\label{E9} \\
&&
 +\langle \nabla_2 F - \nabla_2' F, r_+\nabla_2' H  - r_- \nabla_2 H
   \rangle \nonumber \\
 &&
 +\langle \nabla_1 F,  r_+ \nabla_2' H  - r_-\nabla_2 H  \rangle
- \langle \nabla_1 H,  r_+ \nabla_2' F  - r_- \nabla_2 F \rangle,
\nonumber\eea
where the derivatives are evaluated at $(g,L)$, and we put $rX$ for $r(X)$.
\end{theorem}
\begin{proof}
The first bracket is easily seen to be the Poisson bracket associated with the
canonical symplectic form of the holomorphic cotangent bundle of $G$, which is identified with
$G \times \cG$ using right-translations and the trace form on $\cG$.
The antisymmetry and the Jacobi identity of the second bracket
can be verified by direct calculation.
More conceptually, they follow from the fact that locally, in a neighbourhood of $(\1_n,\1_n)\in G\times  \cG$,
the second bracket can be transformed into Semenov-Tian-Shansky's \cite{STS}
Poisson bracket on the Heisenberg double of the standard Poisson--Lie group
$G$. This is explained in the appendix.
\end{proof}

Let  us display the explicit formula of the Poisson brackets of the evaluation functions given by
the matrix elements $g_{ij}$ and the linear functions $L_a := \langle T_a, L \rangle$ associated with an
arbitrary basis $T_a$ of $\cG$, whose dual basis is $T^a$, $\langle T^b, T_a\rangle = \delta_a^b$.
One may use the standard basis of elementary matrices, $e_{ij}$ defined by $(e_{ij})_{kl} = \delta_{ik} \delta_{jl}$, but we find
it convenient to keep a  general basis.
We obtain directly from the definitions
\be
\nabla g_{ij} = \sum_a T^a (T_a g)_{ij}=  g  e_{ji},
\quad
\nabla' g_{ij} = \sum_a (gT_a)_{ij} T^a = e_{ji} g,
\quad
d L_a = T_a.
\label{ref1}\ee
These give the first Poisson bracket immediately
\be
\{g_{ij}, g_{kl}\}_1 = 0,
\quad
\{ g_{ij}, L_a\}_1 = (T_a g)_{ij},
\quad
\{ L_a, L_b\}_1 = \langle [T_a, T_b], L \rangle.
\label{ref2}\ee
Then elementary calculations lead to the following formulae of the second Poisson bracket,
\be
\{ g_{ij}, g_{kl}\}_2 = \frac{1}{2} \left[ \sgn(i-k) - \sgn(l-j) \right] g_{kj} g_{il},
\label{ref3}\ee
where $\sgn$ is the usual sign function, and
\be
\{ g_{ij}, L_a\}_2 =\Bigl(\bigl( r [T_a, L] + \frac{1}{2} (L T_a + T_a L) \bigr) g\Bigr)_{ij},
\label{ref4}\ee
\be
\{L_a, L_b\}_2 = \langle [L,T_a], r[T_b,L] + \frac{1}{2} (T_b L + L T_b) \rangle.
\label{ref5}\ee
By using the standard basis and evaluating the matrix multiplications,  one may also
spell out the last two equations as

\bea
&& \{ g_{ij},L_{kl}\}= \frac{1}{2}   (\delta_{ik}+\delta_{il}) g_{ij} L_{kl} +
  \delta_{(i>k)} g_{kj} L_{il} +  \delta_{il} \sum_{r>i} L_{kr}g_{rj}\,, \label{R6} \\
&&\{ L_{ij},L_{kl} \}= \frac{1}{2} [\sgn(i-k)+ \sgn(l-j)] L_{il}L_{kj} \nonumber \\
&& \qquad \qquad  \quad \ + \frac{1}{2} (\delta_{il}-\delta_{jk})L_{ij}L_{kl} +
\delta_{il} \sum_{r>i} L_{kr} L_{rj} -   \delta_{jk} \sum_{r> k} L_{ir} L_{rl}\,, \label{R7}
\eea
where  $\delta_{(i>k)} := 1$ if $i>k$ and is zero otherwise.

Let us recall that two Poisson brackets on the same manifold are called
compatible if their arbitrary linear combination is also a Poisson bracket \cite{M}.
Compatible Poisson brackets often arise by taking the Lie derivative
of a given Poisson bracket along a suitable vector field.
If $W$ is a vector field and $\{\ ,\ \}$ is Poisson bracket, then the
Lie derivative bracket is given by
\be
\{ F,H\}^W = W[\{F,H\}] - \{ W[F], H\} - \{ F, W[H] \},
\ee
where $W[F]$ denotes the derivative of the function $F$ along $W$.
This bracket automatically satisfies all the standard properties of a Poisson bracket,
except the Jacobi identity. However, if the Jacobi identity holds for $\{\ ,\ \}^W$, then
$\{\ ,\ \}^W$ and $\{\ ,\ \}$ are  compatible Poisson brackets \cite{FMP,Sm}.

\begin{theorem} \label{Th:2.2+}
The first Poisson bracket of Theorem \ref{Th:2.1} is the Lie derivative of the second Poisson bracket
along the holomorphic vector field, $W$, on $\fM$ whose integral curve through the initial value $(g,L)$ is
\be
\phi_z(g,L) = ( g, L + z \1_n),
\qquad z\in \C,
\label{E10}\ee
where $\1_n$ is the unit matrix. Consequently, the two Poisson brackets are compatible.
\end{theorem}
\begin{proof}
By the general result quoted above  \cite{FMP,Sm}, it is enough to check that
\be
\{ F,H\}_2^W\equiv W[\{F,H\}_2] - \{ W[F], H\}_2 - \{ F, W[H] \}_2 = \{ F,H\}_1
\label{E11}\ee
holds for arbitrary holomorphic functions. Moreover, because of the properties of derivations,
it is sufficient to verify this for the evaluation functions $g_{ij}$ and $L_a$ that yield
coordinates on the manifold $\fM$.
Now it is clear that $W[g_{ij}]=0$ and $W[L_a]$ is a constant.
Therefore, if both $F$ and $H$ are evaluation functions, then $\{ F,H\}_2^W= W[\{F,H\}_2]$.
Thus we see from \eqref{ref3} that the relation $\{ g_{ij}, g_{kl}\}_2^W=0$ is valid.
To proceed further, we use that $W[g]\equiv  \sum_{ij} W[g_{ij}] e_{ij}=0$ and
$W[L] \equiv \sum_a W[L_a] T^a= \1_n$.  Then it follows from the formulae \eqref{ref4} and \eqref{ref5} that
\be
W[\{ g_{ij}, L_a\}_2] = (T_a g)_{ij}
\quad\hbox{and}\quad
W[\{ L_a, L_b\}_2] = \langle [L, T_a], T_b\rangle = \langle L, [T_a, T_b]\rangle.
\ee
Comparison with \eqref{ref2} implies the claim of the theorem.
\end{proof}

\begin{remark} \label{Rem:2.2}
The first line in \eqref{E9} represents the standard multiplicative Poisson structure on the group $G$.
The second line of $\{\ ,\ \}_2$ can be recognized as
the holomophic  extension of the well known Semenov-Tian-Shanksy bracket from $G$ to $\cG$,
where $G$ is regarded as an open submanifold of $\cG$. We recall that
the Semenov-Tian-Shansky bracket originates from the  Poisson--Lie group dual to $G$ \cite{A,STS}.
\end{remark}

Denote by $V_H^i$ $(i=1,2)$ the Hamiltonian vector field associated with the holomorphic function $H$
through the respective Poisson bracket $\{\ ,\ \}_i$.
For any holomorphic function, we have the derivatives
\be
V_H^i[F] = \{ F,H\}_i.
\label{E13}\ee
We are interested in the Hamiltonians
\be
H_m(g,L):= \frac{1}{m} \tr(L^m),
\qquad
\forall m\in \N.
\label{E14}\ee

\begin{proposition} \label{Prop:2.3}
The vector fields associated with the functions $H_m$ are bi-Hamiltonian, since we have
\be
\{ F, H_m\}_2 = \{ F, H_{m+1}\}_1, \qquad \forall m\in \N,\quad \forall F \in \Holm.
\label{E15}\ee
The derivatives of the matrix elements of $(g,L)\in \fM$ give
\be
V_{H_m}^2[g] = V_{H_{m+1}}^1[g] = L^mg,
\qquad
V_{H_m}^2[L] = V_{H_{m+1}}^1[L] =0,
\quad
\forall m\in \N,
\label{E16}\ee
and the flow of $V_{H_m}^2 = V_{H_{m+1}}^1$
through the initial value $(g(0), L(0))$ is
\be
(g(z), L(z)) = (\exp(z L(0)^m) g(0), L(0)).
\label{E17}\ee
\end{proposition}

\begin{proof}
We obtain the derivatives
\be
\nabla_1 H_m(g,L) =\nabla'_1 H_m(g,L) =0,\qquad  d_2 H_{m}(g,L) = L^{m-1},
\quad \forall m=1,2, \dots.
\label{E18}\ee
As a result of \eqref{E7},
\be
\nabla_2 H_m(g,L) = \nabla'_2 H_m(g,L) = L^m,
\ee
and thus, by \eqref{E4},
\be
r_+ \nabla_2' H_m(g, L)  - r_-\nabla_2 H_m(g,L) = L^m = d H_{m+1}(g,L).
\ee
The substitution of these relations into the formulae of Proposition \ref{Th:2.1}  gives
\be
\{F, H_m\}_2(g, L) = \{F, H_{m+1}\}_1(g,L) = \langle \nabla_1 F(g,L), L^m \rangle.
\ee
By the very meaning of the Hamiltonian vector field associated with a function,
these Poisson brackets imply \eqref{E16}, and then \eqref{E17} follows, too.
\end{proof}

Like in the compact case \cite{FeLMP}, we call the $H_m$ `free Hamiltonians' and conclude from
Proposition \ref{Prop:2.3} that they generate a bi-Hamiltonian hierarchy on the holomorphic cotangent bundle $\fM$.

\section{The reduced bi-Hamiltonian hierarchy}
\label{sec:F}

The essence of Hamiltonian symmetry reduction is that one keeps only the `observables' that are
invariant with respect to the pertinent group action.
Here, we  apply this principle to the adjoint action of $G$ on $\fM$, for which $\eta \in G$ acts by the holomorphic
diffeomorphism $A_\eta$,
\be
A_\eta: (g,L) \mapsto (\eta g \eta^{-1}, \eta L \eta^{-1}).
\label{act1}\ee
Thus we keep only the $G$ invariant holomorphic functions
on $\fM$, whose set is denoted
\be
\Holm^G:= \{ F\in \Holm \mid F(g,L) = F(\eta g \eta^{-1}, \eta L \eta^{-1}),\,\, \forall (g,L)\in \fM,\,\, \eta \in G\}.
\label{F1}\ee
For invariant functions, the formula of the second Poisson brackets simplifies drastically.

\begin{lemma} \label{Lem:3.0}
For $F, H\in \Holm^G$, the formula \eqref{E9} can be rewritten as follows:
\be
2\{ F, H\}_2 = \langle \nabla_1 F, \nabla_2 H + \nabla_2' H \rangle - \langle \nabla_1 H, \nabla_2 F + \nabla_2' F \rangle
+ \langle \nabla_2 F, \nabla_2' H \rangle
- \langle \nabla_2 H, \nabla_2' F\rangle.
\label{F1+3}\ee
\end{lemma}

\begin{proof}
We start by noting that for a $G$ invariant function $H$ the
relation
\be
H(g e^{zX}, L) = H(e^{zX} g, e^{zX} L e^{-zX}), \qquad \forall z\in \C,\, \forall X\in \cG,
\label{F1+1}\ee
implies the identity
\be
\nabla'_1 H = \nabla_1 H + \nabla_2 H - \nabla_2' H.
\label{F1+2}\ee
Indeed, since $e^{zX} L e^{-zX} =L+ z X L - z LX + \mathrm{o}(z)$, taking the derivative of both sides of \eqref{F1+1} at $z=0$ gives
\be
\langle X, \nabla_1'H \rangle = \langle X, \nabla_1 H \rangle + \langle XL - LX, d_2 H\rangle = \langle X, \nabla_1 H + \nabla_2 H - \nabla_2' H\rangle.
\label{3ref1}\ee
Since $X$ is arbitrary, \eqref{F1+2} follows.

Formally, \eqref{F1+3} is obtained
from \eqref{E9} by setting $r$ to zero, i.e., $r$ cancels from all terms.
The verification of this cancellation relies on the identity \eqref{F1+2} and is completely straightforward.
We express $\nabla_1' H$ through the other derivatives
with the help of  \eqref{F1+2}, apply the same to $\nabla_1' F$, and then collect terms in \eqref{E9}.
To cancel all  terms containing $r$ we use also that $\langle rX,Y\rangle = - \langle X, rY\rangle$.
 After cancelling those terms,  the equality \eqref{F1+3} is obtained by utilizing the identity
 \be
 \langle \nabla_2 F, \nabla_2 H \rangle - \langle \nabla_2' F, \nabla_2' H \rangle =0,
 \ee
 which is verified by means of the definitions (\ref{E1}) and (\ref{E7}).
\end{proof}

\begin{lemma} \label{Lem:3.1}
$\Holm^G$ is closed with respect to both Poisson brackets of Theorem \ref{Th:2.1}.
\end{lemma}

\begin{proof}
Let us  observe that
the derivatives of the $G$ invariant functions are equivariant,
\be
\nabla_i H (\eta g \eta^{-1}, \eta L \eta^{-1}) = \eta \left(\nabla_i H(g, L))\right) \eta^{-1},
\quad i=1,2,
\label{equivar}\ee
and similar for $\nabla_i' H$. In order to see this, notice that
\be
H(e^{zX} \eta g \eta^{-1}, \eta L \eta^{-1}) = H(\eta^{-1} e^{zX} \eta g, L) = H(e^{z \eta^{-1} X \eta} g, L)
\label{equivar+1}\ee
holds for any $X\in \cG$ and $\eta \in G$ if $H$ is an invariant function. By taking derivative, we obtain
\be
\langle X, \nabla_1 H (\eta g \eta^{-1}, \eta L \eta^{-1}) \rangle =
\langle \eta^{-1} X \eta, \nabla_1 H(g, L) \rangle
=\langle   X , \eta (\nabla_1 H(g, L)) \eta^{-1}  \rangle.
\label{equivar+2}\ee
This leads to the $i=1$ case of \eqref{equivar}.
The property
\be
d_2 H (\eta g \eta^{-1}, \eta L \eta^{-1}) = \eta \left(d_2 H(g, L))\right) \eta^{-1}
\ee
follows in a similar manner, and it implies the  $i=2$ case of \eqref{equivar}.

By combining the formulas \eqref{E8} and \eqref{F1+3} with the equivariance property of the derivatives of $F$ and $H$,
we may conclude from \eqref{trinv}  that if $F, H$ are invariant,  then so is $\{F,H\}_i$ for $i=1,2$.
\end{proof}

We wish to characterize the Poisson algebras of the $G$ invariant functions.
To start, we consider the diagonal subgroup $G_0 < G$,
\be
G_0:= \{ Q  \mid Q=\diag(Q_1,\dots, Q_n),\,\, Q_i \in \C^*\},
\label{F2}\ee
and its regular part $G_0^\reg$, where  $Q_i \neq Q_j$ for all $i\neq j$.
We let $\cN< G$ denote the normalizer of $G_0$ in $G$,
\be
\cN = \{ g\in G\mid g G_0 = G_0 g\}.
\label{F11+}\ee
The normalizer contains $G_0$ as a normal subgroup, and
the corresponding quotient is the permutation group,
\be
\cN/G_0 = S_n.
\label{F3}\ee
We also let $G^\reg\subset G$ denote the dense open subset consisting of the conjugacy classes having representatives in $G_0^\reg$.
Next, we define
\be
\fM^\reg:= \{ (g,L)\in \fM \mid g\in G^\reg\}
\label{F4}\ee
and
\be
\fM_0^\reg:= \{ (Q,L)\in \fM \mid Q\in G_0^\reg \}.
\label{F5}\ee
These are complex manifolds, equipped with their own holomorphic functions.
Now we introduce the chain of commutative algebras
\be
\Holm_\red \subset \Hol(\fM_0^\reg)^{\cN} \subset \Hol(\fM_0^\reg)^{G_0}.
\label{F6}\ee
The last two sets contain the respective invariant elements of $\Hol(\fM_0^\reg)$, and
$\Holm_\red$ \emph{contains the restrictions of the elements of $\Holm^G$ to $\fM_0^\reg$}.
To put this in a more formal manner, let
\be
\iota: \fM_0^\reg \to \fM
\label{F7}\ee
be the tautological embedding.
Then pull-back by $\iota$ provides an isomorphism between  $\Holm^G$ and $\Holm_\red$.
We here used that any holomorphic (or even continuous) function on $\fM$ is
uniquely determined by its restriction to $\fM^\reg$.
Similar, we obtain the map
\be
\iota^*: \Hol(\fM^\reg)^G \to \Hol(\fM_0^\reg)^\cN,
\label{F8}\ee
which is also injective and surjective.

It may be worth elucidating why the pull-back \eqref{F8} is an isomorphism.  To this end, consider
any map $\eta: G^\reg \to G$ such that $\eta(g) g \eta(g)^{-1} \in G_0^\reg$.
Notice that $\eta(g)$ is unique up to left-multiplication by elements of $\cN$ \eqref{F11+}.
Consequently, if $f\in \Hol(\fM_0^\reg)^\cN$, then
\be
F(g,L):= f(\eta(g) g \eta(g)^{-1}, \eta(g) L \eta(g)^{-1})
\label{recon}\ee
yields a well-defined, $G$ invariant function on $\fM^\reg$, which restricts to $f$.
The function $F$ is holomorphic, since locally, on an open set around any fixed
$g_0 \in G^\reg$, one can choose $\eta(g)$ to depend holomorphically on $g$.
Regarding this classical result of perturbation theory, see, e.g.,
Theorem 2.1 in \cite{ACL}.

\begin{definition} \label{Def:3.2}
Let $f,h \in \Holm_\red$ be related to $F,H \in \Holm^G$ by $f = F \circ \iota$ and $h = H \circ \iota$.
In consequence of  Lemma \ref{Lem:3.1}, we can define $\{f, h\}^\red_i \in \Holm_\red$ by the relation
\be
\{ f, h\}_i^\red := \{ F, H\}_i \circ \iota, \qquad i=1,2.
\label{F9}\ee
This gives rise to the reduced Poisson algebras $(\Holm_\red, \{\ ,\ \}_i^\red)$.
\end{definition}

The main goal of this paper is to derive formulae for the reduced Poisson brackets \eqref{F9}.
To do so, we now note that any $f\in \Hol(\fM_0^\reg)$ has the $\cG_0$-valued derivative
$\nabla_1 f$ and the $\cG$-valued derivative $d_2 f$, defined by
\be
\langle \nabla_1 f(Q,L), X_0\rangle = \dz f(e^{zX_0} Q, L),
\quad
\langle d_2 f(Q,L), X \rangle = \dz f(Q, L + z X),
\label{der0}\ee
which are required for all $X_0\in \cG_0$ \eqref{RA}, $X\in \cG$.
For any $Q\in G_0$,  the linear operator $\Ad_Q: \cG \to \cG$ acts as
 $\Ad_Q(X) = Q X Q^{-1}$.
Set
\be
\cG_\perp := \cG_< + \cG_>,
\ee
where $\cG_<$ (resp.~$\cG_>$) is the strictly lower (resp.~upper) triangular subalgebra of $\cG$ introduced in \eqref{RA}.
Notice that for $Q\in G_0^\reg$ the operator $(\Ad_Q - \id)$ maps $\cG_\perp$ to $\cG_\perp$ in an invertible manner.
Building on \eqref{E2}, we have the decomposition
\be
X= X_0 + X_\perp
\quad \hbox{with}\quad
X_\perp = X_< + X_>,\quad \forall X\in \cG.
\ee
Using this, for any $Q\in G_0^\reg$,  the `dynamical $r$-matrix' $\cR(Q) \in \mathrm{End}(\cG)$ is given by
\be
\cR(Q) X = \frac{1}{2} \left( \Ad_Q + \id\right)\circ \left( \Ad_Q - \id\right)_{\vert \cG_\perp}^{-1} X_\perp,
\qquad \forall X \in \cG,
\ee
and we remark its antisymmetry property
\be
\langle \cR(Q) X, Y\rangle = - \langle X, \cR(Q) Y \rangle, \qquad \forall X,Y\in \cG.
\label{asym}\ee
This can be seen  by writing $Q=e^q$ with $q\in \cG_0$ , whereby we obtain
\be
\cR(Q) X = \left(\frac{1}{2} \coth \frac{1}{2} \ad_q \right)  X_\perp,
\label{Rcoth}\ee
Here $\ad_q(X_\perp) = [q, X_\perp]$, which gives an anti-symmetric, invertible linear  operator on $\cG_\perp$.
(The invertibility holds since $Q\in G_0^\reg$, and is needed for $\coth \frac{1}{2}\ad_q$ to be well defined on $\cG_\perp$.)
Below, we shall also employ the shorthand
\be
[X,Y]_{\cR(Q)}:= [ \cR(Q) X, Y] + [X, \cR(Q) Y],
\qquad
\forall X,Y\in \cG.
\label{Rbrac}\ee

\begin{lemma} \label{Lem:3.3}
Consider $f\in \Hol(\fM_0^\reg)^\cN$ given by $f = F \circ \iota$, where $F \in \Hol(\fM^\reg)^G$.
Then the derivatives of $f$ and $F$ satisfy the following relations at any $(Q,L) \in \fM_0^\reg$:
\be
d_2 F(Q,L) = d_2 f(Q,L), \qquad [L, d_2 f(Q,L)]_0 =0,
\label{F14}\ee
\be
\nabla_1 F(Q,L) = \nabla_1 f(Q,L) - (\cR(Q) + \frac{1}{2} \id)[ L, d_2 f(Q,L)].
\label{F15}\ee
\end{lemma}

\begin{proof}
The equalities \eqref{F14} hold since $f$ is the restriction of $F$. In particular, it satisfies
\be
0 = \dz f(Q, e^{z X_0} L e^{-z X_0})  = \langle d_2f(Q,L), [X_0, L]\rangle =
\langle [L, d_2 f(Q,L)], X_0 \rangle, \quad \forall X_0\in \cG_0.
\label{F16}\ee
Concerning  \eqref{F15},  the equality of the $\cG_0$ parts,  $(\nabla_1 F(Q,L))_0 = (\nabla_1 f(Q,L))_0$, is obvious.
Then take any $T \in \cG_\perp$, for which we have
\be
0 = \dz F(e^{z T} Q e^{-zT}, e^{zT} L e^{-zT}) = \langle T, (\id  - \Ad_{Q^{-1}}) \nabla_1 F(Q,L) + [L, d_2 F(Q,L)]\rangle.
\label{F17}\ee
Therefore
\be
(\Ad_{Q^{-1}} - \id ) (\nabla_1 F(Q,L))_\perp = [L, d_2 F(Q,L)]_\perp,
\label{F18}\ee
which implies \eqref{F15}.
\end{proof}

\begin{theorem} \label{Th:3.4}
For  $f,h\in \Hol(\fM)_\reg$, the reduced Poisson brackets defined by \eqref{F9}
 can be described explicitly as follows:
\be
\{f,h\}_1^\red(Q,L) = \langle \nabla_1 f, d_2 h\rangle - \langle \nabla_1 h, d_2 f \rangle
+\langle L, [ d_2 f, d_2 h]_{\cR(Q)} \rangle,
\label{red1}\ee
and
\bea
&&\{f,h\}_2^\red (Q,L)= \frac{1}{2} \langle \nabla_1 f, \nabla_2 h + \nabla_2' h \rangle
- \frac{1}{2} \langle \nabla_1 h, \nabla_2 f + \nabla_2' f\rangle \label{red2} \\
&& \qquad\qquad \qquad \quad +\langle \nabla_2f, \cR(Q)( \nabla_2 h)\rangle  - \langle \nabla_2' f, \cR(Q)( \nabla_2' h)  \rangle,
\nonumber \eea
where all derivatives are taken at $(Q,L)\in \fM_0^\reg$, and the notation \eqref{E7} is in force.
These formulae give two compatible Poisson brackets on $\Hol(\fM)_\red$.
\end{theorem}

\begin{proof}
Let us begin with the first bracket, and note that at $(Q,L)\in \fM_0^\reg$ we have
\be
\langle \nabla_1 F, d_2 H\rangle = \langle \nabla_1 f, d_2 h\rangle
- \langle \cR(Q) [L, d_2 f], d_2 h \rangle - \half \langle [L, d_2 f], d_2 h\rangle,
\ee
since this follows from \eqref{F15}. Now the third term together with the analogous one
coming from $- \langle \nabla_1 H, d_2 F\rangle$ cancel the last term of \eqref{E8}.
Taking advantage of \eqref{asym}, the terms containing $\cR(Q)$ give the
expression written in \eqref{red1}.

Turning to the second bracket, we may start from \eqref{F1+3}, which is valid for elements of $\Holm^G$.
Using \eqref{F15} with
$[L, d_2 f] = \nabla_2 f - \nabla_2' f$, we can write
\bea
&&\langle \nabla_1 F,  \nabla_2' H +  \nabla_2 H\rangle =  \langle \nabla_1 f, \nabla_2' h +  \nabla_2 h \rangle\nonumber\\
  && \qquad \qquad \qquad \qquad  \qquad   + \langle \cR(Q)(\nabla_2' f - \nabla_2 f),  \nabla_2' h + \nabla_2 h \rangle  \label{derred2}\\
 && \qquad \qquad   \qquad \qquad  \qquad  + \frac{1}{2} \langle \nabla_2' f - \nabla_2 f,  \nabla_2' h + \nabla_2 h \rangle .
\nonumber \eea
This holds at $(Q,L)$, since $f, h$  are the restrictions of $F,H\in \Holm^G$.
We then combine \eqref{derred2} with the second term in \eqref{F1+3}.
Collecting terms and using the antisymmetry \eqref{asym}, we obtain
\bea
&&\langle \cR(Q)(\nabla_2' f - \nabla_2 f),  \nabla_2' h + \nabla_2 h \rangle
- \langle \cR(Q)(\nabla_2 h'  - \nabla_2 h),  \nabla_2' f + \nabla_2 f \rangle \nonumber\\
&& \qquad \qquad \qquad = 2 \langle \nabla_2 f, \cR(Q) (\nabla_2 h) \rangle -
2 \langle \nabla_2' f, \cR(Q) (\nabla_2' h) \rangle.
\eea
Moreover, we have
\be
 \frac{1}{2}\langle \nabla_2' f - \nabla_2 f,  \nabla_2' h + \nabla_2 h \rangle  -
 \frac{1}{2}\langle \nabla_2' h - \nabla_2 h,  \nabla_2' f + \nabla_2 f \rangle
 =  \langle \nabla_2' f ,   \nabla_2 h \rangle -     \langle \nabla_2' h ,   \nabla_2 f \rangle,
\ee
which cancels the contribution of the last two terms of \eqref{F1+3}.
In conclusion, we see that
the first and second lines in \eqref{derred2} and their counterparts ensuring antisymmetry  give the claimed
formula \eqref{red2}.

We know from Theorem \ref{Th:2.2+} that the original Poisson brackets on $\Hol(\fM)$ are compatible, which means
that their arbitrary linear combination $\{\ ,\ \}:= x \{\ ,\ \}_1 + y\{\ ,\ \}_2$ satisfies the
Jacobi identity. In particular, the Jacobi identity holds for elements of $\Hol(\fM)^G$ as well.
It is thus plain from Definition \ref{Def:3.2} that the arbitrary linear combination
$\{\ ,\ \}^\red = x \{\ ,\ \}_1^\red + y \{\ ,\ \}_2^\red$ also satisfies the Jacobi identity.
In this way, the compatibility of the two reduced Poisson brackets is
inherited from the compatibility of the original Poisson brackets.
\end{proof}

\begin{remark} \label{Rem:3.5}
It can be shown that the formulae of Theorem \ref{Th:3.4} give Poisson brackets
on $\Hol(\fM_0^\reg)^\cN$ and on $\Hol(\fM_0^\reg)^{G_0}$ as well.
Indeed, we can repeat the reduction starting from $\Hol(\fM^\reg)^G$ using the map
\eqref{F8}, and this leads to the reduced Poisson brackets on $\Hol(\fM_0^\reg)^\cN$.
Then the closure on $\Hol(\fM_0^\reg)^{G_0}$ follows from \eqref{F3} and the local nature of the Poisson brackets.
Because of \eqref{F3}, the quotient by $\cN$ can be taken in two steps,
\be
\fM_0^\reg/\cN = (\fM_0^\reg/G_0)/S_n.
\ee
Since the action of $S_n$ is free, the Poisson structure on  $\fM_0^\reg/\cN$, which
carries the functions $\Hol(\fM_0^\reg)^\cN$, lifts to a Poisson structure on $\fM_0^\reg/G_0$,
whose ring of functions is $\Hol(\fM_0^\reg)^{G_0}$.
\end{remark}

Now we turn to the reduction of the Hamiltonian vector fields \eqref{E16}
to vector fields on $\fM_0^\reg$.  There are two ways to proceed.
One may either directly associate vector fields to the reduced Hamiltonians using
the reduced Poisson brackets, or can suitably `project' the original Hamiltonian
vector fields. Of course, the two methods lead to the same result.

We apply the first method to the reduced Hamiltonians
$h_m:= H_m\circ \iota \in \Hol(\fM)_\red$, which are given by
\be
h_m(Q,L) = \frac{1}{m} \tr (L^m).
\label{F22}\ee
We have to find the vector fields $Y_m^i$ on $\fM_0^\reg$ that satisfy
\be
Y_m^i[f] = \{ f, h_m\}_i^\red,
\quad
\forall f\in \Hol(\fM)_\red,\quad i=1,2.
\label{F23}\ee
These vector fields are not unique, since one may add
any vector field to $Y_m^i$ that is tangent to the orbits of the residual
gauge transformations belonging to the group $G_0$.
This ambiguity does not effect the derivatives of the elements of $\Holm_\red$, and
 we may call any $Y_m^i$ satisfying \eqref{F23} the
\emph{reduced Hamiltonian vector field} associated with $h_m$
and the respective Poisson bracket.

Now a vector field $Y$ on $\fM_0^\reg$ is characterized by the corresponding derivatives
of the evaluation functions that map $\fM_0^\reg\ni (Q,L) $ to $Q$ and $L$, respectively.
 We denote these derivatives by $Y[Q]$ and $Y[L]$.
Then, for any $f\in \Hol(\fM_0^\reg)$,  the chain rule gives
\be
Y[f] = \langle \nabla_1 f, Q^{-1} Y[Q] \rangle + \langle d_2 f, Y[L]\rangle.
\label{F24}\ee

\begin{proposition} \label{Prop:3.6}
For all $m\in \N$, the reduced Hamiltonian vector fields
$Y_m^i$ \eqref{F23} can be specified by the formulae
\be
  Y_{m+1}^1[Q] = Y_m^2 [Q] = (L^m)_0 Q
  \quad \hbox{and}\quad
  Y_{m+1}^1[L] = Y^2_{m}[L] = [ \cR(Q) L^m, L].
 \label{F26} \ee
 \end{proposition}

 \begin{proof}
 It is enough to verify that any $f \in \Hol(\fM)_\red$ and $h_m$ \eqref{F22}, for $m\in \N$, satisfy
\be
\{ f, h_{m+1}\}^1_\red(Q,L) = \{ f, h_m\}_2^\red(Q,L) =
\langle \nabla_1 f(Q,L), (L^m)_0 \rangle + \langle d_2f(Q,L), [\cR(Q) L^m, L] \rangle.
\label{F25}\ee
To obtain this, note that
\be
d_2 h_{m+1}(Q,L) = \nabla_2 h_m(Q,L) = \nabla_2' h_m(Q,L) = L^m.
\ee
Because of \eqref{F14}, these relations of the derivatives reflect those that appeared in the proof of
Proposition \ref{Prop:2.3}.
Putting them into \eqref{red1} gives the claim for $\{f, h_{m+1}\}^1_\red$, since
\be
\langle L, [d_2f(Q,L), L^m]_{\cR(Q)} \rangle = \langle d_2f(Q,L), [\cR(Q)L^m, L]\rangle.
\ee
To get $\{f, h_m\}_2^\red$, we also use that $\nabla_2 f - \nabla_2' f = [L, d_2 f]$. Then the identity
\be
\langle \nabla_2f(Q,L) - \nabla'_2 f(Q,L), \cR(Q)L^m\rangle =
\langle [L, d_2f(Q,L)], \cR(Q)L^m \rangle =
\langle d_2f(Q,L), [\cR(Q)L^m, L]\rangle
\ee
implies \eqref{F25}.
\end{proof}

We conclude from Proposition \ref{Prop:3.6} that the evolutional vector
 fields on $\fM_0^\reg$ that underlie the equations \eqref{I1}
induce commuting bi-Hamiltonian derivations  of the commutative algebra of functions $\Hol(\fM)_\red$.
In this sense, the holomorphic spin Sutherland hierarchy \eqref{I1} possesses a bi-Hamiltonian structure.
It is worth noting that the same statement holds if we replace $\Hol(\fM)_\red$ by  either of the
two spaces of functions in the chain \eqref{F6}. According to \eqref{F8}, $\Hol(\fM_0^\reg)^\cN$
arises by considering the invariants
$\Hol(\fM^\reg)^G$ instead of $\Hol(\fM)^G$. However, it is the latter space that should be regarded
as the proper
algebra of functions on the quotient $\fM/G$ that inherits complete flows from the bi-Hamiltonian  hierarchy on $\fM$.
According to general principles \cite{OR}, the flows on the singular Poisson space $\fM/G$ are just the projections
 of the unreduced flows displayed explicitly in \eqref{E17}.

\section{Recovering the real forms}
\label{sec:R}

It is interesting to see how the bi-Hamiltonian structures of the real forms of the
system \eqref{I1},  described in \cite{FeNon,FeLMP},
 can be recovered from the complex holomorphic case.
First, let us consider the hyperbolic  real form which is obtained by taking $Q$ to be a real, positive matrix,
$Q=e^q$ with a real diagonal matrix $q$,
and $L$ to be a Hermitian matrix.   This means that we replace $\fM_0^\reg$ by the `real slice'
\be
\Re\fM_0^\reg:= \{ (Q,L)\in \fM_0^\reg \mid Q_i=e^{q_i},\, q_i\in \R,\,\, L^\dagger = L\}
\label{R1}\ee
and consider the \emph{real} functions belonging to $C^\infty(\Re\fM_0^\reg)^{\T^n}$, where $\T^n$ is the unitary subgroup of
$G_0$.
For such a function\footnote{We could also consider real-analytic functions.}, say $f$, we can take $\nabla_1 f$ to be a real diagonal matrix
and $d_2 f$ to be a Hermitian matrix.
In fact, in \cite{FeNon} we applied
\be
\langle X, Y \rangle_\R:= \Re \langle X, Y\rangle
\label{R2}\ee
and defined the derivatives  by
\be
\langle \delta q, \nabla_1 f\rangle_\R +
\langle \delta L, d_2 f\rangle_\R := \dt f(e^{t  \delta q}Q, L + t \delta L),
\label{R3}\ee
where $t\in \R$, $\delta q$ is an arbitrary real-diagonal matrix and $\delta L$ is an arbitrary Hermitian matrix.
Notice that the definitions entail
\be
\langle \delta q, \nabla_1 f\rangle_\R +
\langle \delta L, d_2 f\rangle_\R =
\langle \delta q, \nabla_1 f\rangle+
\langle \delta L, d_2 f\rangle,
\ee
and, with $\nabla_2 f \equiv L d_2 f$,
\be
\nabla'_2 f  \equiv (d_2 f) L = (L d_2 f)^\dagger = (\nabla_2 f)^\dagger.
\label{R4}\ee

\begin{proposition} \label{Prop:4.1}
If we consider $f, h \in C^\infty(\Re\fM_0^\reg)^{\T^n}$ with \eqref{R1} and insert their derivatives as defined above
into the right-hand sides of the formulae of Theorem \ref{Th:3.4}, then we obtain the following real Poisson brackets:
\be
\{ f,h\}^\Re_1(Q,L) = \langle  \nabla_1 f,  d_2 h \rangle_\R
- \langle  \nabla_1 h,  d_2 f\rangle_\R
+\langle L, [ d_2 f,  d_2 h]_{\cR(Q)}   \rangle_\R,
\label{R5}\ee
and
\be
\{ f,h\}_2^\Re(Q,L) = \langle  \nabla_1 f, \nabla_2 h\rangle_\R
- \langle  \nabla_1 h,  \nabla_2 f\rangle_\R
+ 2\langle   \nabla_2 f ,  \cR(Q)(  \nabla_2 h)\rangle_\R,
\label{R6}\ee
which reproduce the real bi-Hamiltonian structure given in Theorem 1 of \cite{FeNon}.
\end{proposition}

\begin{proof}
The proof relies on the identity
\be
\cR(Q) (X^\dagger) = - (\cR(Q) X)^\dagger,
\quad
\forall X\in \cG.
\label{R8}\ee
This can be seen, for example, from the  formula \eqref{Rcoth}, since
\be
\ad_q X^\dagger = [q, X^\dagger] = - [q, X]^\dagger= - (\ad_q X)^\dagger, \qquad \forall X\in \cG,
\label{R9}\ee
because in the present case $q$ is a real diagonal matrix.
To deal with the first bracket, note that $\langle \nabla_1 f, d_2 h\rangle= \langle \nabla_1 f, d_2 h\rangle_\R$ as
both $\nabla_1 f$ and $d_2 h$ are Hermitian.
By using \eqref{R8} and the definition \eqref{Rbrac}, we see that $[ d_2 f, d_2 h]_{\cR(Q)}$ is Hermitian as well, and thus
\be
\langle L, [ d_2 f, d_2 h]_{\cR(Q)} \rangle =\langle L, [ d_2 f, d_2 h]_{\cR(Q)} \rangle_\R.
\ee
Consequently, we obtain the formula \eqref{R5} from \eqref{red1}

Turning to the second bracket, the equality $\nabla_2' h = (\nabla_2 h)^\dagger$  \eqref{R4} implies
\be
\frac{1}{2} \langle \nabla_1 f, \nabla_2 h + \nabla_2' h \rangle =
\frac{1}{2} \langle \nabla_1 f, \nabla_2 h \rangle  + \frac{1}{2} \langle (\nabla_1 f)^\dagger , (\nabla_2 h)^\dagger \rangle =
\langle \nabla_1 f, \nabla_2 h\rangle_\R,
\label{R7}\ee
simply because $\langle X, Y \rangle^* = \langle X^\dagger,  Y^\dagger \rangle$ holds for all $X,Y\in \cG$.
Thus the first line of \eqref{red2} correctly gives the first two terms of \eqref{R6}.
Moreover, on account of \eqref{R4} and \eqref{R8}, we obtain
\bea
&&\langle \nabla_2f, \cR(Q)( \nabla_2 h)\rangle  - \langle \nabla_2' f, \cR(Q)( \nabla_2' h)  \rangle \\
&& \quad =\langle \nabla_2f, \cR(Q)( \nabla_2 h)\rangle + \langle (\nabla_2f)^\dagger, (\cR(Q)( \nabla_2 h))^\dagger\rangle
= 2\langle \nabla_2f, \cR(Q)( \nabla_2 h)\rangle_\R.\nonumber
\eea
Therefore, \eqref{red2} gives \eqref{R6}.

Comparison with Theorem 1 in \cite{FeNon} shows
that the formulae \eqref{R5} and \eqref{R6} reproduce the real bi-Hamiltonian structure derived in that paper. We remark
that our $d_2 f$  \eqref{R3}  was denoted $\nabla_2 f$, and our variable $q$ corresponds to $2q$ in \cite{FeNon}.
Taking this into account, the Poisson brackets of Proposition \ref{Prop:4.1}, multiplied by an overall factor 2, give
precisely the Poisson brackets of \cite{FeNon}.
\end{proof}

The real form treated above yields the hyperbolic spin Sutherland model, and now we deal with the
trigonometric case. For this purpose, we introduce the alternative real slice
 \be
\Re'\fM_0^\reg:= \{ (Q,L)\in \fM_0^\reg \mid Q_j=e^{\ri q_j},\, q_j\in \R,\,\, L^\dagger = L\}
\label{R18}\ee
and consider the \emph{real} functions belonging to $C^\infty(\Re'\fM_0^\reg)^{\T^n}$.
A bi-Hamiltonian structure on this space of functions was derived in \cite{FeLMP}, where we used
the pairing
\be
\langle X, Y \rangle_\I:= \Im \langle X, Y\rangle
\label{R19}\ee
and defined the derivatives $D_1 f$, which is a real diagonal matrix, and $D_2 f$, which is an anti-Hermitian matrix, by
the requirement
\be
\langle \ri \delta q, D_1 f\rangle_\I +
\langle \delta L, D_2 f\rangle_\I := \dt f(e^{t \ri \delta q}Q, L + t \delta L),
\label{R20}\ee
where $t\in \R$,  $\delta q$ is  an arbitrary real-diagonal matrix and $\delta L$ is an arbitrary Hermitian matrix.
It is readily seen that
\be
\langle \ri \delta q, D_1 f\rangle_\I +
\langle \delta L, D_2 f\rangle_\I = \langle \ri \delta q, -\ri D_1 f\rangle +
\langle \delta L, -\ri D_2 f\rangle,
\ee
and comparison with \eqref{E5} motivates the definitions
\be
\nabla_1 f:= - \ri D_1 f,
\quad
d_2 f:= - \ri D_2 f.
\label{R21}\ee
This implies that $\nabla_2 f := L d_2 f$ and $\nabla_2' f := (d_2 f) L$ satisfy \eqref{R4} in this case as well.
An important difference is that instead of \eqref{R8} in the present case we have
\be
\cR(Q) X^\dagger = (\cR(Q) X)^\dagger,
\qquad \forall X\in \cG,
\label{R22}\ee
because  in \eqref{Rcoth} $q$ gets replaced by $\ri q$ with a real $q$, and then instead of \eqref{R9} we have
$\ad_{\ri q} X^\dagger = (\ad_{\ri q} X)^\dagger$.

\begin{proposition} \label{Prop:4.2}
If we consider $f, h \in C^\infty(\Re'\fM_0^\reg)^{\T^n}$ with \eqref{R18}  and insert their derivatives as defined in \eqref{R21}
into the right-hand sides of the formulae of Theorem \ref{Th:3.4}, then we obtain the following purely imaginary Poisson brackets:
\be
\{ f,h\}^\I_1(Q,L) = -\ri \bigl(\langle  D_1 f,  D_2 h \rangle_\I
- \langle  D_1 h,  D_2 f\rangle_\I
+\langle L, [ D_2 f,  D_2 h]_{\cR(Q)}   \rangle_\I \bigr),
\label{R23}\ee
and
\be
\{ f,h\}_2^\I(Q,L) = -\ri \bigl( \langle  D_1 f, L D_2 h\rangle_\I
- \langle  D_1 h,  L D_2 f\rangle_\I
+ 2\langle   L D_2 f ,  \cR(Q)( L  D_2 h)\rangle_\I \bigr).
\label{R24}\ee
Then $\ri \{ f,h\}_1^\I$ and $\ri \{f,h\}_2^\I$ reproduce
the real bi-Hamiltonian structure
given in Theorem 4.5 of \cite{FeLMP}.
\end{proposition}

\begin{proof}
We detail only the first bracket, for which the first term of \eqref{red1} gives
\be
\langle \nabla_1 f, d_2 h\rangle = - \langle D_1 f, D_2 h\rangle = -\ri \langle D_1 f, D_2 h\rangle_\I,
\label{R25}\ee
since $\langle D_1 f, D_2 h\rangle$ is purely imaginary. The second term of \eqref{red2} is similar, and
the third term gives
\be
\langle L, [d_2 f, d_2 h]_{\cR(Q)} \rangle = - \langle L, [D_2 f, D_2 h]_{\cR(Q)} \rangle = - \ri \langle L, [D_2 f, D_2 h]_{\cR(Q)} \rangle_\I,
\label{R26}\ee
since  $[D_2 f, D_2 h]_{\cR(Q)}$ is anti-Hermitian.
To see this, we use \eqref{Rbrac} noting that $D_2 f$ and, by \eqref{R22}, $\cR(Q)(D_2f)$ are anti-Hermitian (and the same for $h$).
Collecting terms, the formula \eqref{R23} is obtained. The proof of \eqref{R24} is analogous to the calculation
presented in the proof of \eqref{R6}. The difference arises from the fact that now we have  \eqref{R22} instead of \eqref{R8}.
The last statement of the proposition is a matter of obvious comparison with the formulae of Theorem 4.5 of \cite{FeLMP}
(but one should note that what we here call $D_2 f$ was denoted $d_2 f$  in that paper, and $\langle\ ,\ \rangle_\I$ was
denoted $\langle\ ,\ \rangle$).
\end{proof}

\section{Conclusion}
\label{sec:C}

In this paper we developed a bi-Hamiltonian interpretation for the system of
holomorphic evolution equations \eqref{I1}.
The bi-Hamiltonian structure was found by interpreting
this hierarchy  as the Poisson reduction of a bi-Hamiltonian hierarchy on the
holomorphic cotangent bundle $T^*\GL$, described by Theorems \ref{Th:2.1}, \ref{Th:2.2+} and
Proposition \ref{Prop:2.3}.
Our main result is given by Theorem \ref{Th:3.4} together with
Proposition \ref{Prop:3.6}, which characterize the reduced bi-Hamiltonian hierarchy.
Then we reproduced our previous results on real forms
of the system \cite{FeNon,FeLMP}
by considering real slices of the holomorphic reduced phase space.

The first reduced Poisson structure and the associated interpretation
as a spin Sutherland model is well known, and it is also known that
the restrictions of the system to generic symplectic leaves of $T^*\GL/\GL$
are integrable in the degenerate sense \cite{Res1}.
 Experience with the real forms \cite{FeLMP} indicates that the second Poisson structure
 should be tied in with a relation of  the reduced system to spin
 Ruijsenaars--Schneider models, and degenerate integrability should
 also hold on the corresponding symplectic leaves.
 We plan to come back to this issue elsewhere.
 We remark in passing that although
  $T^*\GL/\GL$ is not a manifold, this does not cause any serious difficulty,
  since it still can be decomposed as a disjoint union of symplectic leaves.
  This follows from  general results on singular Hamiltonian reduction \cite{OR}.

 We finish by highlighting a few open problems for future work.
 First, it could be interesting to explore degenerate integrability directly
 on the  Poisson space $T^* \GL/\GL$, suitably adapting the formalism of
 the paper \cite{LMV}.
 Second, we wish to gain a better conceptual understanding of the process whereby
 one goes from holomorphic Poisson spaces and integrable systems to their real forms,
 and apply it to our case. The results of the recent study \cite{AF} should be
 relevant in this respect. Finally, it is a challenge to generalize our construction
 from the hyperbolic/trigonometric case to elliptic systems.
The existence of a bi-Hamiltonian structure for the elliptic spin Calogero--Moser system
 appears to follow from the existence of such a structure for an integrable
 elliptic top on $\GL$ \cite{KLO} via the symplectic  Hecke correspondence \cite{LOV,O}.

\bigskip
\bigskip
\begin{acknowledgements}
I wish to thank Maxime Fairon for several useful remarks on the manuscript.
I am also grateful to Mikhail Olshanetsky for drawing my attention to relevant references.
This work was supported by the NKFIH research grant K134946, and was also supported partially
by the Project GINOP-2.3.2-15-2016-00036 co-financed by the European Regional Development
Fund and the budget of Hungary.
\end{acknowledgements}

\appendix

\section{The origin of the second Poisson bracket on $G\times \cG$}

In this appendix we outline how the Poisson bracket $\{\ ,\ \}_2$ \eqref{E9} arises from the standard
Poisson bracket \cite{STS} on the Heisenberg double of the $\GL$ Poisson--Lie group.

We start with the complex Lie group  $G\times G$ and denote its elements as
pairs $(g_1, g_2)$.
We equip the corresponding Lie algebra $\cG \oplus \cG$ with the nondegenerate bilinear form
$\langle\ , \ \rangle_2$, given by
\be
\langle (X_1,X_2), (Y_1, Y_2) \rangle_2 :=
\langle X_1, Y_1 \rangle - \langle X_2, Y_2 \rangle
\label{G1}\ee
for all $(X_1,X_2)$ and $(Y_1,Y_2)$ from $\cG \oplus \cG$.
Then we have the  isotropic subalgebras,
\be
\cG^\delta := \{ (X,X)\mid X \in \cG\},
\label{G2}\ee
and
\be
\cG^*:= \{ (r_+(X), r_-(X))  \mid \forall X  \in \cG\}.
\label{G3}\ee
Recall that $r_\pm$ are defined in \eqref{rpm}, and  note that
$\cG \oplus \cG$ is the vector space direct sum of the disjunct subspaces $\cG^\delta$ and $\cG^*$;
$\cG^\delta$ is isomorphic to $\cG$, and $\cG^*$ can be regarded as its linear dual space.
We also introduce the corresponding subgroups of $G\times G$,
\be
G^\delta := \{ g_\delta  \mid
g_\delta:= (g,g),\, g \in G\},
\label{G4}\ee
and
\be
G^*= \left\{  g_* \mid g_*:=\left(g_> g_0, ( g_0 g_<)^{-1}\right),\,  g_> \in G_>,\, g_0\in G_0,\, g_< \in G_<\right\},
\label{G5}\ee
where $G_>$, $G_<$ and $G_0$ are the connected subgroups of $G$ associated with the Lie subalgebras
in the decomposition \eqref{RA}.
That is, $G_0$ contains the diagonal, invertible complex matrices, and $G_>$ (resp. $G_<$) consists of
the upper triangular (resp. lower triangular)
complex matrices whose diagonal entries are all equal to $1$.

In order to describe the pertinent  Poisson structures, we need the
 Lie algebra valued derivatives of holomorphic functions. For $\cF\in \Hol(G\times G)$,
we denote its $\cG \oplus \cG$-valued left- and right-derivatives, respectively, by $\cD\cF$ and $\cD' \cF$.
For example, we have
\be
\langle (X_1, X_2), \cD \cF(g_1,g_2)\rangle_2: = \dz \cF(e^{z X_1} g_1, e^{z X_2} g_2),
\label{G6}\ee
where $z\in \C$ and $(X_1,X_2)$ runs over $\cG \oplus \cG$.
Defined using $\langle\ ,\ \rangle_2$, a holomorphic function $\phi$ on $G^\delta$ has the $\cG^*$-valued
left- and right-derivatives, $D \phi$ and $D'\phi$.
Analogously, the left- and right-derivatives $D\chi$ and $D'\chi$ of $\chi\in \Hol(G^*)$ are
$\cG^\delta$-valued.

Now we recall \cite{STS} that $G\times G$ carries two natural Poisson brackets, which are given by
\be
\{\cF, \cH\}_\pm := \langle \cD \cF, R \cD \cH \rangle_ 2 \pm  \langle \cD' \cF, R \cD' \cH \rangle_ 2,
\label{PBpm}\ee
 where $R:= \frac{1}{2} \left( P_{\cG^\delta} - P_{\cG^*}\right)$ with the projections
 $P_{\cG^\delta}$ onto $\cG^\delta$ and $P_{\cG^*}$  onto $\cG^*$
 defined via the vector space direct sum $\cG \oplus \cG = \cG^\delta + \cG^*$.
 The minus bracket is called the Drinfeld double bracket, and the plus one the Heisenberg double bracket.
 The former makes $G\times G$ into a Poisson--Lie group, having the Poisson submanifolds $G^\delta$ and $G^*$,
 and the latter is symplectic in a neighbourhood of the identity.

 Let us consider an open neighbourhood of the identity in $G\times G$ whose elements can be factorized as
 \be
 (g_1, g_2) = g_{\delta L} g_{* R}^{-1} = g_{* L} g_{\delta R}^{-1}
 \label{G8}\ee
 with $g_{\delta L}, g_{\delta R}\in G^\delta$ and $g_{*L}, g_{*R}\in G^*$.
 Restricting $(g_1, g_2)$ as well as all constituents in the factorizations to be near enough
to the respective identity elements, the map
\be
(g_1, g_2) \mapsto (g_{\delta R}, g_{*R})
\label{G9} \ee
 yields a local, biholomorphic diffeomorphism.
As the first step towards deriving the bracket in \eqref{E9},
we use this diffeomorphism to transfer the plus Poisson bracket to a neighbourhood of the
identity of $G^\delta \times G^*$. The resulting Poisson structure then extends holomorphically
to the full of $G^\delta \times G^*$.
For $\cG, \cH \in \Hol(G^\delta \times G^*)$ we denote the resulting Poisson bracket
by $\{\cF,\cH\}_+'$.  One can verify that it takes  the following form:
  \bea
&&\{\cF, \cH\}_+'(g_\delta,g_*) =\left\langle g_* (D_2' \cF) g_*^{-1}, D_2\cH  \right\rangle_2
-\left\langle  g_\delta ( D'_1\cF) g_\delta^{-1},  D_1\cH \right\rangle_2
\nonumber\\
&&\qquad \qquad  +  \left\langle D_1\cF , D_2\cH \right\rangle_2
-\left\langle D_1 \cH , D_2\cF \right\rangle_2,
\label{+PB1}\eea
where the derivatives on the right-hand side
are   taken at $(g_\delta,g_*)\in G^\delta\times G^*$.
The subscript $1$ and $2$ refer to  derivatives with respect to the first and second arguments;
they are $\cG^*$ and $\cG^\delta$ valued, respectively.
For example, we have
\be
\langle D_1 \cF(g_\delta, g_*), (X,X)\rangle_2 = \dz \cF( (e^{zX}, e^{zX}) g_\delta, g_*)
\label{D1def}\ee
and
\be
\langle D_2 \cF(g_\delta, g_*), (r_+X,r_-X)\rangle_2 = \dz \cF(  g_\delta,(e^{zr_+X}, e^{zr_-X}) g_*).
\label{D2def}\ee
It is worth noting that
\be
(r_+X, r_-X) = (X_> + \frac{1}{2} X_0, - X_< - \frac{1}{2} X_0) \quad\hbox{for}\quad X= (X_> + X_0 + X_<) \in \cG.
\label{rpmX}\ee
The derivatives $D_1'$ and $D_2'$ are defined analogously, cf. \eqref{E5}.
The derivation of the formula \eqref{+PB1} from $\{\ ,\ \}_+$ in \eqref{PBpm} can follow closely the
proof of Proposition 2.1 in \cite{FeLMP},
where another Heisenberg double was treated.
The formula \eqref{+PB1} itself has the same structure as formula (2.18) in \cite{FeLMP},
and thus we  here omit  its derivation.

In the second step towards getting $\{\ ,\ \}_2$ in \eqref{E9}, we make use of a
biholomorphic diffeomorphism between open neighbourhoods
of the identity element of $G^\delta \times G^*$ and the element $(\1_n, \1_n) \in G\times \cG$.
For $g_\delta = (g,g)$ and $g_* = (g_> g_0, (g_0 g_<)^{-1} )$, this is given by the map
\be
(g_\delta, g_*) \mapsto (g, L)
\quad \hbox{with}\quad
L:= g_> g_0^2 g_<.
\label{G11}\ee
A (locally defined) function $\cF$ on $G^\delta \times G^*$   then corresponds to a (locally defined) function $F$
on $G\times \cG$ according to
\be
\cF(g_\delta, g_*) \equiv F(g,L).
\label{G12}\ee
To proceed further, we need an auxiliary result.

\begin{lemma} \label{Lm:A.1}
For the functions $\cF$ and $F$ in \eqref{G12},
 the derivatives $D_i \cF$
 and $D_i' \cF$ $(i=1,2)$ defined in \eqref{D1def}, \eqref{D2def} and the derivatives $\nabla_i F$, $\nabla_i' F$ defined in \eqref{E5}--\eqref{E7}
 are related as follows:
 \bea
&& D_1 \cF(g_\delta, g_*) = ( r_+ \nabla_1 F(g,L), r_- \nabla_1 F(g,L)), \label{G13} \\
&& D_1' \cF(g_\delta, g_*) = ( r_+ \nabla_1' F(g,L), r_- \nabla_1' F(g,L)),\label{G13*}
 \eea
 \be
D_2 \cF(g_\delta, g_*) = ( r_+ \nabla_2' F(g,L) - r_- \nabla_2 F(g,L),   r_+ \nabla_2' F(g,L) - r_- \nabla_2 F(g,L)),
\label{G14}\ee
and
\be
 P_{\cG^*}\left( g_* D_2' \cF(g_\delta, g_*) g_*^{-1} \right) = P_{\cG^*}\left((\nabla_2 F(g,L), \nabla_2'F(g,L))\right).
\label{G15}\ee
\end{lemma}
\begin{proof}
We begin by pointing out the identity
\be
\langle (r_+Y, r_-Y), (X,X) \rangle_2 = \langle Y, X\rangle, \quad \forall X,Y\in \cG,
\label{P1}\ee
which is a consequence of \eqref{G1} and \eqref{rpmX}.
Now the definitions of the derivatives ensure that
\be
 \langle D_1 \cF(g_\delta, g_*), (X,X) \rangle_2 = \langle \nabla_1 F(g,L), X \rangle, \qquad \forall X \in \cG.
 \label{P2}\ee
Because of \eqref{P1} and the non-degeneracy of both pairings, this implies the identity \eqref{G13},
and   \eqref{G13*} results in the same manner.

To derive \eqref{G14}, we may forget the $g_\delta$-dependence and assume (just for simplicity of writing)
that $\cF$ depends only on $g_*$, which we now write as
\be
g_* = (g_+, g_-) \quad \hbox{with}\quad g_+ = g_> g_0,\,\, g_- = (g_0 g_<)^{-1},
\label{P3}\ee
referring to \eqref{G5}.
By setting $\hat X:= (r_+X, r_-X)$ and writing $D$ for $D_2$, we have
\be
\langle \hat X, D\cF(g_*) \rangle_2= \dz \cF(e^{ z r_+ X} g_+, e^{z r_- X} g_-) = \dz F (e^{z r_+X} L e^{-z r_- X})
\ee
since $L= g_+ g_-^{-1}$.
By simply expanding the exponential functions, this is equal  to
\bea
&&\langle  dF(L), (r_+X) L - L r_-X \rangle = \langle \nabla F(L), r_+X \rangle - \langle \nabla' F(L), r_- X \rangle
\nonumber\\
&&
= \langle r_+ \nabla' F(L) - r_- \nabla F(L), X \rangle
\nonumber \\
&& =\langle \left(r_+ \nabla' F(L) - r_- \nabla F(X),  r_+ \nabla' F(L) - r_- \nabla F(L)\right),  \hat X \rangle_2.
\eea
To get this, we used the definitions \eqref{E1}, \eqref{E7} together with
the anti-symmetry of $r$ \eqref{E3} with respect to the trace form, and the identity \eqref{P1}.
Thus we have shown that
\be \langle \hat X, D\cF(g_*) \rangle_2 = \langle \hat X,
 \left(r_+ \nabla' F(L) - r_- \nabla F(X),  r_+ \nabla' F(L) - r_- \nabla F(L)\right) \rangle_2
 \ee
 for arbitrary $\hat X \in \cG^*$ \eqref{G3}.
This implies \eqref{G14} since $\cG^\delta$ \eqref{G2} and $\cG^*$ \eqref{G3} are in duality with
respect to the non-degenerate pairing $\langle\ ,\ \rangle_2$.

In order to derive \eqref{G15}, we again assume that $\cF$ depends only on $g_* = (g_+, g_-)$.
Then we note that, for any $V\in \cG$,
\be
\langle P_{\cG^*}\left( g_* D' \cF( g_*) g_*^{-1} \right), (V,V)\rangle_2 = \langle g_* D' \cF( g_*) g_*^{-1} , (V,V)\rangle_2=
\langle D' \cF( g_*),  P_{\cG^*}  \left(g_*^{-1} (V,V) g_*\right)\rangle_2.
\ee
Of course, now  $D' \cF(g_*) \equiv D'_2 \cF(g_*) \in \cG_\delta$.
If we consider the decomposition
\be
g_*^{-1} (V,V) g_* = (K_+, K_-) + (U,U), \qquad
(K_+, K_-):=  P_{\cG^*}  \left(g_*^{-1} (V,V) g_*\right)
\ee
then
\be
K_+ - K_- = g_+^{-1} V g_+ - g_-^{-1} V g_-.
\ee
By using these we can write
\bea
&&\langle D' \cF( g_*),  P_{\cG^*}  \left(g_*^{-1} (V,V) g_*\right)\rangle_2 = \dz \cF (g_+ e^{K_+ z}, g_- e^{K_-z} )\nonumber\\
&& = \dz F(g_+ e^{K_+z} e^{- K_-z} g_-^{-1}) = \dz F(L + g_+( (K_+ - K_-) z + {\mathrm{o}}(z)) g_-^{-1})\nonumber\\
&& = \dz F(L+ (VL - LV)z + {\mathrm{o}}(z))  = \langle dF(L), VL - LV \rangle   \nonumber \\
&&= \langle (LdF(L),  (dF(L)) L), (V,V)\rangle_2.
\eea
 Hence we have shown that
 \be
\langle g_* D' \cF( g_*) g_*^{-1} , (V,V)\rangle_2 = \langle (\nabla F(L), \nabla' F(L)), (V,V)\rangle_2
\ee
which implies the claimed formula.
\end{proof}

We apply the local diffeomorphism \eqref{G11} to transfer the Poisson bracket $\{\ ,\ \}_+'$ \eqref{+PB1} to a Poisson bracket of holomorphic
functions defined locally on $G \times \cG$ (that is, on an open subset containing $(\1_n, \1_n) \in G\times \cG$).
The formula of the transferred Poisson bracket is obtained by substituting the identities of Lemma \ref{Lm:A.1}
into \eqref{+PB1} and then simply collecting terms.
The result turns out to have the form $\{\ ,\ \}_2$ given in \eqref{E9}, and  it naturally extends to a globally well defined Poisson bracket
 of holomorphic functions on $\fM = G \times \cG$.


\begin{thebibliography}{99}

\bibitem{ACL}
A.L. Andrew, K.-W. E. Chu and P. Lancaster,
{\it Derivatives of eigenvalues and eigenvectors of matrix functions},
SIAM J. Matrix  Anal. Appl. {\bf 14} (1993) 903-926

\bibitem{AAJ}
I. Aniceto, J. Avan and A. Jevicki,
{\it Poisson structures of Calogero--Moser and Ruijsenaars--Schneider models},
J. Phys. A {\bf 43} (2010) 185201;
\href{https://arxiv.org/abs/0912.3468}{\tt  arXiv:0912.3468 [hep-th]}

\bibitem{AF}
P. Arathoon and M. Fontaine,
{\it Real forms of holomorphic Hamiltonian systems};
\href{https://arxiv.org/abs/2009.10417}{\tt  arXiv:2009.10417 [math.SG]}

 \bibitem{A}
 G. Arutyunov,
Elements of Classical and Quantum Integrable Systems, Springer, 2019

\bibitem{BBT}
O. Babelon, D. Bernard, and M. Talon, Introduction to Classical Integrable Systems,
 Cambridge University Press,  2003.


\bibitem{BDF}
J. Balog, L. D\c{a}browski and L. Feh\'er,
{\it Classical $r$-matrix and exchange algebra in WZNW and Toda theories},
Phys. Lett. B {\bf 244} (1990) 227-234

\bibitem{Bar}
C. Bartocci, G. Falqui, I. Mencattini, G. Ortenzi and M. Pedroni, {\it On the geometric
origin of the bi-Hamiltonian structure of the
Calogero--Moser system}, Int. Math. Res. Not.  {\bf 2010} 279-296;
\href{https://arxiv.org/abs/0902.0953}{\tt arXiv:0902.0953 [math-ph]}

\bibitem{DeKV}
A. De Sole, V.G. Kac and  D. Valeri,
{\it Classical affine W-algebras and the associated integrable Hamiltonian hierarchies for classical Lie algebras},
 	Comm. Math. Phys. {\bf 360} (2018) 851-918; 
 \href{https://arxiv.org/abs/1705.10103}{\tt arXiv:1705.10103 [math-ph]}


 \bibitem{EV}
 P. Etingof and A. Varchenko,
 {\it  Geometry and classification of solutions of the classical dynamical Yang--Baxter equation},
Commun. Math. Phys. {\bf 192} (1998) 77-120;
 \href{https://arxiv.org/abs/q-alg/9703040}{\tt arXiv:q-alg/9703040}


\bibitem{FMP}
G. Falqui, F. Magri and M. Pedroni,
{\it Bihamiltonian geometry, Darboux coverings, and linearization of the KP hierarchy},
Commun. Math. Phys. {\bf 197} (1998) 303-324;  
\href{https://arxiv.org/abs/solv-int/9806002}{\tt arXiv:solv-int/9806002}


\bibitem{FaMe}
G. Falqui and I. Mencattini,
{\it Bi-Hamiltonian geometry and canonical spectral coordinates for
the rational Calogero--Moser system}, J. Geom. Phys. {\bf 118} (2017) 126-137;
\href{https://arxiv.org/abs/1511.06339}{\tt arXiv:1511.06339 [math-ph]}


\bibitem{FeNon}
L. Feh\'er,
{\it Bi-Hamiltonian structure of a dynamical system introduced by Braden and Hone},
Nonlinearity {\bf 32} (2019) 4377-4394;
\href{https://arxiv.org/abs/1901.03558}{\tt arXiv:1901.03558 [math-ph]}

\bibitem{FeLMP}
L. Feh\'er,
{\it Reduction of a bi-Hamiltonian hierarchy on $T^*{\mathrm{U}}(n)$ to spin Ruijsenaars--Sutherland models},
Lett. Math. Phys. {\bf 110} (2020) 1057-1079;  
\href{https://arxiv.org/abs/1908.02467}{\tt arXiv:1908.02467 [math-ph]}


\bibitem{KLO}
B. Khesin, A. Levin and M. Olshanetsky,
{\it Bihamiltonian structures and quadratic algebras in hydrodynamics and on non-commutative torus},
Commun. Math. Phys. {\bf 250} (2004) 581-612;
\href{https://arxiv.org/abs/nlin/0309017}{\tt arXiv:nlin/0309017 [nlin.SI]}

\bibitem{LMV}
C. Laurent-Gengoux, E. Miranda and P. Vanhaecke,
{\it Action-angle coordinates for integrable systems on Poisson manifolds},
Int. Math. Res. Not. {\bf  2011} 1839-1869;  	
\href{https://arxiv.org/abs/0805.1679}{\tt arXiv:0805.1679 [math.SG]}

\bibitem{LOV}
A.M. Levin, M.A. Olshanetsky and A. Zotov,
{\it Hitchin systems\,--\,symplectic Hecke correspondence and two-dimensional version},
Commun. Math. Phys. {\bf 236} (2003) 93–133; 
\href{https://arxiv.org/abs/nlin/0110045}{\tt  arXiv:nlin/0110045 [nlin.SI]}


\bibitem{LX}
L.-C. Li and P. Xu,
{\it A class of integrable spin Calogero--Moser systems},
Commun. Math. Phys. {\bf 231} (2002) 257-286;
\href{https://arxiv.org/abs/math/0105162}{\tt arXiv:math/0105162 [math.QA]}


\bibitem{M}
F. Magri,
{\it A simple model of the integrable Hamiltonian equation},
 J. Math. Phys. {\bf 19}  (1978) 1156-1162

\bibitem{O}
M. Olshanetsky,
{\it Classical integrable systems and gauge field theories},
Phys. Part. Nucl. {\bf 40} (2009) 93–114;
\href{https://arxiv.org/abs/0802.3857}{\tt arXiv:0802.3857  [hep-th]}

\bibitem{OR}
J.-P. Ortega and T. Ratiu,
Momentum Maps and Hamiltonian Reduction, Birk\"auser, 2004



\bibitem{Res1}
N. Reshetikhin,
{\it Degenerate integrability of spin Calogero--Moser systems and the
duality with the spin Ruijsenaars systems},
Lett. Math. Phys. {\bf 63} (2003) 55-71;
\href{https://arxiv.org/abs/math/0202245}{\tt arXiv:math/0202245 [math.QA]}


\bibitem{RBanff}
S.N.M. Ruijsenaars,
{\it Systems of Calogero--Moser type}, pp. 251-352 in: Proceedings of the
1994 CRM-Banff Summer School: Particles and Fields, Springer, 1999

\bibitem{STS}
  M.A. Semenov-Tian-Shansky,
{\it Dressing transformations and Poisson group actions},
Publ. RIMS {\bf 21} (1985) 1237-1260


\bibitem{Sm}
R.G. Smirnov,
{\it Bi-Hamiltonian formalism: A constructive approach},
Lett. Math. Phys. {\bf 41} (1997)  333-347

\bibitem{SurB}
Yu. B. Suris,
The Problem of Integrable Discretization:
Hamiltonian Approach,
Birkh\"auser, 2003

\end{thebibliography}
\end{document}